\documentclass{article}

\newenvironment{proof}[1][Proof]{\begin{trivlist}
\item[\hskip \labelsep {\bfseries #1}]}{\end{trivlist}}

\newcommand{\qed}{\nobreak \ifvmode \relax \else
      \ifdim\lastskip<1.5em \hskip-\lastskip
      \hskip1.5em plus0em minus0.5em \fi \nobreak
      \vrule height0.75em width0.5em depth0.25em\fi}

\usepackage{natbib}

\usepackage[colorinlistoftodos]{todonotes}
\setcitestyle{authoryear,open={(},close={)}}

\usepackage{amsmath}
\newtheorem{prop}{Proposition}
\newtheorem{exmp}{Example}
\usepackage[normalem]{ulem}
\usepackage{color}

\usepackage{authblk}

\begin{document}

\title{Why Scientists Chase Big Problems: Individual Strategy and Social Optimality\thanks{All authors contributed equally to this paper. This work was supported by a generous grant to the Metaknowledge Network from the John Templeton Foundation.}}
\author{Carl T. Bergstrom}
\affil{Department of Biology, University of Washington}
\author{Jacob G. Foster}
\affil{Department of Sociology, UCLA}
\author{Yangbo Song}
\affil{Departments of Economics and Sociology, UCLA; School of Management and Economics, Chinese University of Hong Kong (Shenzhen)}
\date{\today}

\maketitle

\begin{abstract}
Scientists pursue collective knowledge, but they also seek personal recognition from their peers. When scientists decide whether or not to work on a big new problem, they weigh the potential rewards of a major discovery against the costs of setting aside other projects. These self-interested decisions may distribute researchers across problems in an efficient manner, but efficiency is not guaranteed. We use simple economic models to understand such decisions and their collective consequences. Academic science differs from industrial R\&D in that academics often share partial solutions to gain reputation. This convention of Open Science is thought to accelerate collective discovery, but we find that it need not do so. The ability to share partial results influences which scientists work on a particular problem; consequently, Open Science can slow down the solution of a problem if it deters entry by important actors.
\end{abstract}

\section{Introduction}

In the development of a scientific field, problems arise that are widely recognized as important. Their solutions may be critical to progress in the field, valuable to society, or both. Faced with such a problem, researchers in the field must make a choice: should they continue pursuing their workaday lines of inquiry, or should they join the race to solve the problem? This is a strategic decision, since an individual scientist's best choice depends on the choices of other researchers. A problem worth pursuing in the absence of competition may be a bad bet if several top research teams are racing for the solution.

This decision has both personal and societal consequences. An individual scientist can obtain considerable prestige by solving a recognized important problem, but she will have wasted precious research effort should she be beaten to the solution by a competitor. For problems with significant social impact, society can benefit when many scientists join the race and as a result the problem is rapidly solved. On the other hand, it can be wasteful if too many scientists pursue a single problem in parallel, because too much effort is diverted away from other research activity.

There is no \emph{a priori} reason why strategic individual behavior should generate socially optimal outcomes. If we want to understand the conditions under which scientific institutions support more or less efficient allocations of cognitive labor \citep{kitcher1990division}, we need a model in which aggregate behavior arises from individual allocation decisions---decisions made in response to incentives, given preferences over time and risk. This approach naturally falls into a domain sometimes dubbed the ``new'' economics of science \citep{partha1994toward,stephan1996economics}. Building on early contributions by seminal figures such as  \citet{nelson1959economics} and \citet{arrow1962economic}, this research program emphasizes the analysis of distinctive scientific institutions; draws explicitly on classic contributions from the sociology of science \citep{merton1957priorities,merton1968matthew,merton1973sociology,latour1987science}; and uses a variety of game theoretic models to analyze distinct steps in the production of scientific and technical knowledge, from problem choice \citep{dasgupta1987} to the timing of research \citep{dasgupta1988} to information sharing \citep{dasgupta1987information}. 

We draw here on two other streams of literature, in addition to the new economics of science discussed above. One stream comes from the philosophy of science, and the other comes from the microeconomic study of patenting. The philosophy stream has its roots in early work by Peirce (1897)\nocite{peirce1897}, and tries to understand how apparently idiosyncratic features of the scientific credit systems lead to more or less efficient allocations of scientists to problems \citep{kitcher1995advancement, strevens2003role, strevens2006role, strevens2011economic}. 
Within this stream, researchers have used a variety of formal
models to analyze distinct steps in the production of scientific and technical
knowledge, from problem choice \citep{kleinberg2011mechanisms} to the timing of research, to team formation \citep{bikard2015exploring,boyer2015scientific}, to information sharing norms \citep{zollman2009optimal}. 

Here we build specifically on \cite{kleinberg2011mechanisms}. These authors find that seemingly unfair aspects of credit allocation like the Matthew effect \citep{merton1968matthew} can support an efficient distribution of scientists across problems. However, they also demonstrate that there are generic cases in which the credit system possesses insufficient degrees of freedom to impel socially efficient outcomes. But their model is atemporal; thus it cannot be used for treating decisions about {\em when} to publish, nor can it be readily extended to account for time preferences on the part of scientists or of planners. As we will see, accounting for time preferences will be important in a nuanced analysis of social welfare and efficiency. Moreover, heterogeneity in scientists' time preferences can further reduce the prospects for socially optimal allocations. 

By contrast, the relevant microeconomic literature on patenting tackles the issue of time directly. Early economic models of patent races explore the R\&D investment decisions that firms make in an effort to secure intellectual property rights \citep{loury1979market,lee1980market}; see \cite{stiglitz2014creating} for a contemporary treatment. In these analyses, the patent race is typically modeled as a symmetric game without {\em ex ante} heterogeneity among participants in competitive ability, opportunity cost, or time preferences. These modeling decision assure a unique equilibrium.

Building on the initial patent models, later authors focus on the effect of partial information disclosure. A firm may disclose partial information to signal its strength and commitment to the race \citep{anton2003,gill2008}, to increase the entire market demand \citep{defraja1993}, to induce exit of risk-averse competitors \citep{bhattacharya1983,lichtman2000,baker2005} or to establish prior art as a defensive measure \citep{parchomovsky2000}. 

Like the contributors to this literature, we are concerned with identifying incentives for sharing valuable information with potential competitors. Because we are modeling the decisions of academic scientists, however, we allow the disclosure of partial results to be rewarding \textit{per se}: disclosure confers recognition, which is a foundation of academic science \citep{merton1957priorities,merton1973sociology}. 

In the past two years, researchers have begun to model 
how academic priority races influence decisions about whether to disclose partial results. The basic tension is the same as we explore here: by publishing partial results, a researcher or research team is assured some credit for the progress already made, at the cost of ceding any lead in the  race to solve the bigger problem. 
\citet{boyer2014bird} lays out a basic model of this scenario, and finds that in some circumstances the risk of losing the larger race to a major finding creates incentives for researchers to settle for the bird-in-hand, publishing intermediate results and claiming what credit they can.   
\citet{banerjee2014re} consider a more complicated model of scientific progress as a ``treasure hunt'' game in which multiple alternate paths may lead to the end result. Like Boyer (2014), they find that researchers will often choose to obtain partial credit by publishing intermediate findings along the way. \citet{heesen2015communism} finds a clever way to depict multi-stage credit races as an extensive-form game of incomplete information; under certain symmetry assumptions, sharing all partial results is a unique Nash equilibrium in his model.

Our analysis diverges from these treatments because we provide the researchers in our model with an outside option: we allow them to opt in or out of the focal problem. This is important for a several reasons. First, by making participation a strategic decision, our model highlights the strategic dependence of one agent's payoffs for opting in on the decisions of other agents about whether to participate as well \footnote{Oren and Kleinberg are able to capture a similar phenomenon by developing a model in which researchers can choose among multiple alternative projects.}. Second, by allowing scientists to opt in or opt out of the focal problem, we are able to make explicit social welfare comparisons with regard to the under- or over-allocation of scientists to given problems. Third, our approach captures strategic complexities that recommend against maximal information sharing. When everyone \emph{has} to opt in, sharing obviously accelerates the solution of the full problem. Information sharing becomes strategically interesting once individuals can opt out; different sharing rules have different consequences for who participates in each stage. For example, when scientists are able to share partial results, certain individuals might drop out of the race once the easy stages have been completed, whereas they would have stayed to the end if only the complete solution had been publishable. Thus it is far from trivial whether information sharing will accelerate overall progress on the focal problem. 

In this paper, we describe a formal model of scientific decision-making by scientists who act as self-interested strategic agents rather than idealized seekers of knowledge\footnote{This approach shares much with the picture arising from ethnographic and historical work within science studies \citep{bourdieu1975specificity,panofsky2014misbehaving, latour1987science}. By connecting formal methods with such qualitative analyses, we can provide science studies with formal microfoundations, and ground the modeling work of economists and philosophers in the finely textured observation of practicing scientists.}. In this model, scientists choose whether to forego the rewards of their everyday research and enter a race to solve an important focal problem. We show that this game has a unique risk-dominant equilibrium under reasonable assumptions and demonstrate that the existence of this equilibrium is robust to heterogeneity among scientists in scientific capital. We then extend the model to allow the publication of partial results. We provide a framework for teasing apart scientists' incentives from society's outcomes. We demonstrate that the institution of Open Science (i.e., the publication of partial results) is not universally beneficial, because it can slow the time to solution for important problems and thereby decrease social welfare. 

\section{Model setup}
Our primary goal is to understand scientists' decision to participate in a race to solve a commonly acknowledged, widely valued problem $\mathcal{P}$. We hence consider the set of all scientists $I=\{1,2,\cdots,N\}$ that could conceivably work on this focal problem. Time $t$ is continuous. At $t=0$ the existence and importance of the problem becomes common knowledge. In practice, this big problem could arise in many ways: the development of an enabling technique or technology (e.g., CRISPR/Cas9 gene editing technology allowing us to develop gene drive systems to eliminate mosquito malaria vectors; \cite{hammond2016crispr}); the sudden emergence of a socially-relevant challenge (e.g., the AIDS epidemic of the 1980's and the need to identify its etiological agent; \cite{epstein1996impure}); or the pressure to resolve contradictory lines of evidence or an embarrassment of accumulating anomalies \citep{kuhn2012structure}, with examples ranging from the role of persistent antigen in preserving immunological memory \citep{antia1998models} to the Lamb shift in the energy levels of hydrogen, which stimulated the development of quantum electrodynamics \citep{schweber1994qed}. In the simplest case, we assume that scientists have a single opportunity to opt in; if a scientist does not participate at $t=0$, she cannot participate at any point later in time (we will relax this assumption subsequently). To participate, scientists must make an initial outlay $F>0$, representing the time needed to master any relevant techniques and/or the funds required to equip or staff the lab to pursue problem $\mathcal{P}$. We assume that this outlay is the same across scientists.

We will model scientists' decisions to work on problems that, while universally acknowledged as important, do not have a known method of solution and hence require more than rote effort. We capture this feature of the problem through the distribution of the time to solution $\tau_i$. We assume that $\tau_i$ follows an exponential distribution\footnote{Huber (1998, 2001; cf. Heesen 2015) \nocite{huber1998invention,huber2001new} shows that patents and papers are distributed across a career in a manner consistent with an exponential distribution of times-to-solution.}:
\begin{align*}	
\mbox{Prob}(\tau_i\leq t)=1-e^{-\frac{h_i}{d}t}.
\end{align*}

Scientists differ in their intrinsic ability to make progress on the focal problem and in the amount of capital (material, financial, human, social) they can bring to bear. We capture this difference with an scientist-specific parameter $h_i$ quantifying $i$'s scientific capital. Problems also vary in their intrinsic difficulty; we capture this variation in the difficulty of the focal problem with a parameter $d$. A scientist's instantaneous probability of solving the focal problem is thus $\frac{h_i}{d}$. All else being equal, problems with higher $d$ will take longer to solve, and scientists with lower $h_i$ will take longer to solve them. We assume common knowledge of these parameters.

Under the priority rule, the first scientist to solve the problem gets a fixed payoff $V>0$ (but see \cite{merton1968matthew,strevens2003role,strevens2006role}) and her unsuccessful competitors in the race get nothing. Racing to solve the focal problem diverts researchers from their usual work. Therefore, participants in the race incur an opportunity cost of $x$ per unit of scientific capital per unit time, i.e., if researcher $i$ works on $\mathcal{P}$ for time $t$ she will pay an opportunity cost $h_i\,x\,t$. We assume in this simplest case that $x$ is the same across all scientists and that there is no time-discount of payoffs, although we will relax both of these assumptions subsequently.

If a scientist chooses not to participate, she gets a payoff of $0$. A scientist will only benefit from participation if her instantaneous return from participation is positive, i.e., $V\, \frac{h_i}{ d}  - h_i\,x > 0$; but because of the fixed cost this not a sufficient condition for participation. Given this, continuing to participate until the problem is solved strictly dominates any strategy with an earlier exit, because if one stops early, one gets an instantaneous return of 0. Therefore we only consider strategies in which scientists persist until a problem is solved.  Now we can easily compute the expected payoff of any scientist $i$ who does participate. Let $I'\subset I$ denote the set of participating scientists. Then $h_{I'} = \sum_{j \in I'} h_j$ is the total scientific capital of all participating scientists and $h_{I'-i}=\sum_{j\in I',j\neq i}h_j$ is the total scientific capital of all of $i$'s competitors. Given an exponential distribution of solution times, $i$ solves the problem with probability $\frac{h_i}{h_{I'}}$ and the expected time to solution is $\frac{d}{h_{I'}}$. Thus, the expected payoff to $i$ of participation is
\begin{align*}
U_i(I',h_i,h_{I'-i})=V\frac{h_i}{h_{I'}} - h_i\,x\frac{d}{h_{I'}} - F
\end{align*}

We define an \textbf{equilibrium} of this game as a subset of participating scientists $I^*\subset I$ such that
\begin{align*}
U_i(I^*,h_i,h_{I^*-i})&\geq 0\text{ for every $i\in I^*$}\\
U_j(I^*+j,h_j,h_{I^*})&\leq 0\text{ for every $j\notin I^*$}
\end{align*}
This is a Nash equilibrium, i.e., no scientist can unilaterally improve her expected payoff by altering her decision whether or not to participate. We will refer to the model described here as the {\em basic game.}

\section{Results}

\subsection{Multiple Equilibria and Risk Dominance}

As in typical patent race models, a Nash equilibrium is guaranteed to exist. Unlike much of the patent race literature, our assumption of {\em ex ante} variation in scientific capital means that there are generically multiple equilibria even in our most basic model. We first prove the existence of an equilibrium.

\begin{prop}
A Nash equilibrium always exists in the basic game.
\end{prop}

\begin{proof}
For notational convenience, we construct the index set $I=\{1,2,\cdots,N\}$ such that scientists are labeled in descending order of scientific capital $h_i$: $h_1\geq h_2\geq\cdots\geq h_N$. Now let $I^*=\{1,\cdots,i^*\}$ where $i^*=\max\{i:(\frac{V}{d}-x)\frac{d\,h_i}{\sum_{j=1}^i h_j}-F\geq 0\}$, and consider $I^*$ as the participating group of scientists. Note that for each scientist $i<i^*$, her expected payoff from participation is strictly higher than that of scientist $i^*$, and hence strictly positive; also by construction of $i^*$, if any scientist $i>i^*$ chose to participate, her expected payoff would be strictly negative. Therefore, we know that $I^*$ is an equilibrium. This completes the proof.
\end{proof}

Moreover, there are commonly multiple equilibria in this game.

\begin{exmp}
Consider three paleolinguists: Alice, Bob, and Carol. Renowned scholar Alice has scientific capital $h_a=10$, whereas mid-career colleague Bob has scientific capital $h_b=5$, and early career researcher Carol has $h_c=4$ respectively. A sequence of symbols from the upper Paleolithic is discovered in an Australian cave, and the race is on to decode its meaning. The value of finding the solution is $V=20$, the difficulty of the problem is $d=5$ and the opportunity cost of working on the problem is $x=1$ per unit scientific capital per unit time. In addition, working on the problem incurs an immediate fixed cost of $F=4$ for each scientist.

If Alice and Bob work on the problem, it is not worth it for Carol to enter because her expected payoff from participation would be negative. If Alice and Carol work on the problem, Bob will stay out for a similar reason. If Bob and Carol were to work on the problem, Alice would \emph{not} stay out; her expected payoff from participation would be positive. If Alice joins, it would behoove either Bob or Carol to drop, because either one's expected payoff from participation would be negative given that the other stayed. Therefore this game has two Nash equilibria. In one, the participants are $\{Alice, Bob\}$ and in the other $\{Alice, Carol\}$.
\end{exmp}

At first glance, multiple equilibria are problematic: they undercut the predictive value of the theory. However, intuition suggests that the Alice-Bob equilibrium is more plausible than the Alice-Carol equilibrium. After all, if Carol mistakenly enters the Alice-Bob equilibrium, the cost to Carol would be higher than the cost to Bob if Bob mistakenly enters the Alice-Carol equilibrium; this is because Carol has less scientific capital than Bob. The same intuition applies to accidental exit by the scientists participating in an equilibrium: if Bob mistakenly exits the Alice-Bob equilibrium, the expected loss to Bob (i.e., the expected payoff that Bob could have earned) would be higher than the expected loss to Carol if Carol mistakenly exits the Alice-Carol equilibrium.

For two-player games, economists formalize this kind of intuition with the concept of risk dominance \citep{harsanyi1988general}. We extend the notion of risk dominance to $n$ player games (cf. \nocite{kojima2006risk} Kojima 2006). Let $I^*\setminus\hat{I}$ be the set of scientists in $I^*$ but not $\hat{I}$. We define an equilibrium $I^*$ as \textit{risk dominant} if for any equilibrium $\hat{I}$, the following inequality holds:
\begin{align*}
&|\sum_{i\in I^*\setminus\hat{I}}U_i(I^*,h_i,h_{I^*-i})\times\sum_{j\in \hat{I}\setminus I^*}U_j(I^*\cup \hat{I},h_j,h_{I^*\cup \hat{I}-j})|\\
\geq &|\sum_{j\in \hat{I}\setminus I^*}U_j(\hat{I},h_j,h_{\hat{I}-j})\times\sum_{i\in I^*\setminus \hat{I}}U_i(I^*\cup \hat{I},h_i,h_{I^*\cup \hat{I}-i})|
\end{align*}

The first term on the left-hand side, $\sum_{i\in I^*\setminus\hat{I}}U_i(I^*,h_i,h_{I^*-i})$, is the total loss of the subgroup $I^*\setminus\hat{I}$ should they play their strategies in $\hat{I}$ instead of their strategies in $I^*$. In our game, this is the total loss to  $I^*\setminus\hat{I}$ if they fail to participate when they should have. The second term $\sum_{j\in \hat{I}\setminus I^*}U_j(I^*\cup \hat{I},h_j,h_{I^*\cup \hat{I}-j})$ is the total loss of the subgroup $\hat{I}\setminus I^*$ should they play their strategies in $\hat{I}$ instead of their strategies in $I^*$. In our game, this is the total loss to $\hat{I}\setminus I^*$ if they participate when they should have stayed out. Note that in both cases the strategies of the other scientists are held constant. The right-hand side of the inequality can be interpreted in a similar way by flipping the two equilibria.

We focus on generic cases in which the $h_i$'s are different from one another. To make the problem non-trivial, we assume that $V-dx-F>0$. To avoid excessive technicality, we assume that a scientist participates when she is indifferent between participation and non-participation. We thus prove the following proposition:

\begin{prop}
A risk dominant equilibrium exists and is unique: $I^*=\{1,\cdots,i^*\}$ where $i^*=\max\{i:V\frac{h_i}{\sum_{j=1}^ih_j} - h_i\,x\frac{d}{\sum_{j=1}^ih_j} - F\geq 0\}$.
\end{prop}

\begin{proof}
From the assumption that $V-dx-F>0$ we know that $I^*$ is not empty; by construction of $i^*$ we know that $I^*$ is an equilibrium.

Now consider any other equilibrium $\hat{I}\neq I^*$. Again by construction of $i^*$, we know that both $I^*\setminus\hat{I}$ and $\hat{I}\setminus I^*$ are non-empty. Let $I^*\setminus\hat{I}=\{i_1,\cdots,i_m\}$ and $\hat{I}\setminus I^*=\{j_1,\cdots,j_n\}$, and re-label the scientists such that $h_{i_1}>h_{i_2}>\cdots>h_{i_m}$ and $h_{j_1}>h_{j_2}>\cdots>h_{j_n}$. It immediately follows that $i_m\geq i^*>j_1$. Also, for $\hat{I}$ to be an equilibrium, it must be the case that $m\leq n$.

For notational convenience, let $A=\sum_{i\in I^*}h_i$, $B=\sum_{j\in \hat{I}}h_j$, and $C=A+\sum_{p=1}^n h_{j_p}=B+\sum_{q=1}^m h_{i_q}$. Let $l=|I^*\cap \hat{I}|$. We discuss two cases: $A\geq B$ and $A<B$.

When $A\geq B$: to show that $I^*$ is the unique risk dominant equilibrium, it suffices to show that
\begin{align}
\sum_{i\in I^*\setminus\hat{I}}U_i(I^*,h_i,h_{I^*-i})&>\sum_{j\in \hat{I}\setminus I^*}U_j(\hat{I},h_j,h_{\hat{I}-j})\\
\sum_{j\in \hat{I}\setminus I^*}U_j(I^*\cup \hat{I},h_j,h_{I^*\cup \hat{I}-j})&<\sum_{i\in I^*\setminus \hat{I}}U_i(I^*\cup \hat{I},h_i,h_{I^*\cup \hat{I}-i}).
\end{align}
Note that in $(2)$, the term on either side has a negative value.

To prove $(1)$, we compute the total expected payoff in the two equilibria:
\begin{align*}
\sum_{i\in I^*}U_i(I^*,h_i,h_{-i})&=V-dx-(m+l)F\\
\sum_{j\in \hat{I}}U_j(\hat{I},h_j,h_{-j})&=V-dx-(n+l)F.
\end{align*}
From $m\leq n$, we have
\begin{align}
\sum_{i\in I^*}U_i(I^*,h_i,h_{I^*-i})> \sum_{j\in \hat{I}}U_j(\hat{I},h_j,h_{\hat{I}-j}).
\end{align}
Now, notice that
\begin{align*}
\sum_{k\in I^*\cap\hat{I}}U_k(I^*,h_k,h_{I^*-k})&=(\frac{V}{d}-x)\frac{d\sum_{k\in I^*\cap\hat{I}}h_k}{A}-l\,F\\
\sum_{k\in I^*\cap\hat{I}}U_k(\hat{I},h_k,h_{\hat{I}-k})&=(\frac{V}{d}-x)\frac{d\sum_{k\in I^*\cap\hat{I}}h_k}{B}-l\,F.
\end{align*}
From $A\geq B$, we have
\begin{align}
\sum_{k\in I^*\cap\hat{I}}U_k(I^*,h_k,h_{I^*-k})\leq \sum_{k\in I^*\cap\hat{I}}U_k(\hat{I},h_k,h_{I^*-k}).
\end{align}
Combining $(3)$ and $(4)$, and noticing that $m\leq n$ and $A\geq B$ cannot be binding at the same time (i.e., at least one of them must be strict), we have the desired inequality $(1)$.

To prove $(2)$, we write the terms on both sides as
\begin{align*}
\sum_{i\in I^*}U_i(I^*\cup \hat{I},h_i,h_{I^*\cup \hat{I}-i})&=(\frac{V}{d}-x)\frac{d\,A}{C}-(m+l)F\\
\sum_{j\in \hat{I}}U_j(I^*\cup \hat{I},h_j,h_{I^*\cup \hat{I}-j})&=(\frac{V}{d}-x)\frac{d\,B}{C}-(n+l)F.
\end{align*}
Following a similar argument to above, we have
\begin{align}
\sum_{i\in I^*}U_i(I^*\cup \hat{I},h_i,h_{I^*\cup \hat{I}-i})> \sum_{j\in \hat{I}}U_j(I^*\cup \hat{I},h_j,h_{I^*\cup \hat{I}-j}).
\end{align}
Subtracting the term on either side of $(5)$ from the total expected payoff in $I^*\cup \hat{I}$ (which is not an equilibrium), we have the desired inequality $(2)$.

When $A<B$: to show that $I^*$ is the unique risk dominant equilibrium, it suffices to show that
\begin{align}
\frac{\sum_{i\in I^*\setminus\hat{I}}U_i(I^*,h_i,h_{I^*-i})}{m}&>\frac{\sum_{j\in \hat{I}\setminus I^*}U_j(\hat{I},h_j,h_{\hat{I}-j})}{n}\\
\frac{\sum_{j\in \hat{I}\setminus I^*}U_j(I^*\cup \hat{I},h_j,h_{I^*\cup \hat{I}-j})}{n}&<\frac{\sum_{i\in I^*\setminus \hat{I}}U_i(I^*\cup \hat{I},h_i,h_{I^*\cup \hat{I}-i})}{m}.
\end{align}
Note that in $(7)$, the term on either side has a negative value.

Let $D=\sum_{q=1}^m h_{i_q}$ and $E=\sum_{p=1}^n h_{j_p}$. The above inequalities can be written as
\begin{align*}
(\frac{V}{d}-x)\frac{D}{m\,A}-F&>(\frac{V}{d}-x)\frac{E}{m\,B}-F\\
(\frac{V}{d}-x)\frac{E}{n\,C}-F&<(\frac{V}{d}-x)\frac{D}{m\,C}-F,
\end{align*}
which are further equivalent to
\begin{align}
\frac{D}{m}\frac{1}{A}&>\frac{E}{n}\frac{1}{B}\\
\frac{E}{n}&<\frac{D}{m},
\end{align}
From the assumption $A<B$, it immediately follows that $\frac{1}{A}>\frac{1}{B}$. Also, $\frac{E}{n}<\frac{D}{m}$ is guaranteed by $i_m\geq i^*>j_1$ derived before. Therefore, $(8)$ and $(9)$ are satisfied, which implies $(6)$ and $(7)$. This completes the proof. 
\end{proof}

Thus we expect to see realized the equilibrium in which the researchers with the highest scientific capital pursue the problem and all other researches stay out of the race. We demonstrate this by explicit calculation for the previous example.

\begin{exmp}
Consider the two equilibria in Example 1. Let $I^*=\{Alice, Bob\}$ and $\hat{I}=\{Alice, Carol\}$. It follows that $I^*\setminus \hat{I}=\{Bob\}$ and $\hat{I}\setminus I^*=\{Carol\}$. Then we can calculate the corresponding expected payoffs in the inequality characterizing risk dominance:
\begin{align*}
&|\sum_{i\in I^*\setminus\hat{I}}U_i(I^*,h_i,h_{I^*-i})\times\sum_{j\in \hat{I}\setminus I^*}U_j(I^*\cup \hat{I},h_j,h_{I^*\cup \hat{I}-j})|=|1\times (-0.84)|\\
\geq &|\sum_{j\in \hat{I}\setminus I^*}U_j(\hat{I},h_j,h_{\hat{I}-j})\times\sum_{i\in I^*\setminus \hat{I}}U_i(I^*\cup \hat{I},h_i,h_{I^*\cup \hat{I}-i})|=|0.29\times (-0.05)|
\end{align*}
Hence, $\{Alice, Bob\}$ is the unique risk dominant equilibrium.
\end{exmp}

\subsection{Heterogeneous outside options}

The uniqueness of the risk-dominant equilibrium does not require that all players have the same opportunity cost per unit of scientific capital per unit time.
We demonstrate this as follows. Let the per-unit opportunity cost $x_i$ be heterogeneous, in the sense that it is a function of $h_i$: $x_i=\chi(h_i)$. In other words, scientists in our model now have heterogeneous outside options.The total opportunity cost per unit time $h_i \chi(h_i)$ is still assumed to be increasing in $h_i$. The expected payoff of scientist $i$ becomes $V\frac{h_i}{h_{I'}}-h_i \chi(h_i) \frac{d}{h_{I'}}-F$. To make the problem non-trivial, assume that $V-\chi(h_1)\,d-F>0$.
We prove the following proposition.

\begin{prop}
Let $I^*=\{1,\cdots,i^*\}$ where $i^*=\max\{i:V\frac{h_i}{\sum_{j=1}^i h_j}-h_i\,\chi(h_i)\frac{d}{\sum_{j=1}^i h_j}-F\geq 0\}$. $I^*$ is the unique risk dominant equilibrium if $h\,\chi(h)$ is concave in $h$ and $\frac{h_{i^*-1}\,\chi(h_{i^*-1})-h_{i^*}\,\chi(h_{i^*})}{h_{i^*-1}-h_{i^*}}<\frac{V}{d}$.
\end{prop}

\begin{proof}
When $h\,\chi(h)$ is concave in $h$ and $\frac{h_{i^*-1}\,\chi(h_{i^*-1})-h_{i^*}\,\chi(h_{i^*})}{h_{i^*-1}-h_{i^*}}<\frac{V}{d}$, the expected payoff to scientist $i\in I^*$ is decreasing in $i$. Hence $I^*$ is indeed an equilibrium.

To prove that $I^*$ is the unique risk dominant equilibrium, we adopt the argument in the proof of Proposition 2 using the same notation. Consider any other equilibrium $\hat{I}$. It still follows that $i_m\geq i^*>j_1$ and $m\leq n$. As before, we discuss two cases: $A\geq B$ and $A<B$.

When $A\geq B$: we aim to prove the same inequalities $(1)$ and $(2)$. To prove $(1)$, let $a=\frac{\sum_{i\in I^*}h_i\,\chi(h_i)}{A}$ and $b=\frac{\sum_{j\in\hat{I}}h_j\,\chi(h_j)}{B}$. From the assumption that $h\,\chi(h)$ is concave, we know that $a<b$. We then have $\sum_{i\in I^*}h_i\,\chi(h_i)=a\,A$ and $\sum_{j\in \hat{I}}h_j\,\chi(h_j)=b\,B$. The total expected payoffs in the two equilibria are given by
\begin{align*}
\sum_{i\in I^*}U_i(I^*,h_i,h_{I^*-i})&=(V\frac{A}{d}-\sum_{i\in I^*}h_i\,\chi(h_i))\frac{d}{A}-(m+l)F=V-a\,d-(m+l)F\\
\sum_{j\in \hat{I}}U_j(\hat{I},h_j,h_{\hat{I}-j})&=(V\frac{B}{d}-\sum_{j\in\hat{I}}h_j\,\chi(h_j))\frac{d}{B}-(n+l)F=V-b\,d-(n+l)F.
\end{align*}
From $m\leq n$, we have
\begin{align}
\sum_{i\in I^*}U_i(I^*,h_i,h_{I^*-i})> \sum_{j\in \hat{I}}U_j(\hat{I},h_j,h_{\hat{I}-j}).
\end{align}
Now, notice that
\begin{align*}
\sum_{k\in I^*\cap\hat{I}}U_k(I^*,h_k,h_{I^*-k})&=(V\frac{\sum_{k\in I^*\cap\hat{I}}h_k}{d}-\sum_{k\in I^*\cap\hat{I}}h_k\,\chi(h_k))\frac{d}{A}-l\,F\\
\sum_{k\in I^*\cap\hat{I}}U_k(\hat{I},h_k,h_{\hat{I}-k})&=(V\frac{\sum_{k\in I^*\cap\hat{I}}h_k}{d}-\sum_{k\in I^*\cap\hat{I}}h_k\,\chi(h_k))\frac{d}{B}-l\,F.
\end{align*}
From $A\geq B$, we have
\begin{align}
\sum_{k\in I^*\cap\hat{I}}U_k(I^*,h_k,h_{I^*-k})\leq \sum_{k\in I^*\cap\hat{I}}U_k(\hat{I},h_k,h_{\hat{I}-k}).
\end{align}
Combining $(10)$ and $(11)$, we have the desired inequality $(1)$.

To prove $(2)$, we write the terms on both sides as
\begin{align*}
\sum_{i\in I^*}U_i(I^*\cup \hat{I},h_i,h_{I^*\cup \hat{I}-i})&=(V\frac{A}{d}-a\,A)\frac{d}{C}-(m+l)F\\
\sum_{j\in \hat{I}}U_j(I^*\cup \hat{I},h_j,h_{I^*\cup \hat{I}-j})&=(V\frac{B}{d}-b\,B)\frac{d}{C}-(n+l)F.
\end{align*}
Following a similar argument to above, we have
\begin{align}
\sum_{i\in I^*}U_i(I^*\cup \hat{I},h_i,h_{I^*\cup \hat{I}-i})> \sum_{j\in \hat{I}}U_j(I^*\cup \hat{I},h_j,h_{I^*\cup \hat{I}-j}).
\end{align}
Subtracting the term on either side of $(5)$ from the total expected payoff in $I^*\cup \hat{I}$ (which is not an equilibrium), we have the desired inequality $(2)$.

When $A<B$: to show that $I^*$ is the unique risk dominant equilibrium, it suffices to show that
\begin{align}
\frac{\sum_{i\in I^*\setminus\hat{I}}U_i(I^*,h_i,h_{I^*-i})}{m}&>\frac{\sum_{j\in \hat{I}\setminus I^*}U_j(\hat{I},h_j,h_{\hat{I}-j})}{n}\\
\frac{\sum_{j\in \hat{I}\setminus I^*}U_j(I^*\cup \hat{I},h_j,h_{I^*\cup \hat{I}-j})}{n}&<\frac{\sum_{i\in I^*\setminus \hat{I}}U_i(I^*\cup \hat{I},h_i,h_{I^*\cup \hat{I}-i})}{m}.
\end{align}
Let $a'=\frac{\sum_{i\in I^*\setminus\hat{I}}h_i\,\chi(h_i)}{D}$ and $b'=\frac{\sum_{j\in \hat{I}\setminus I^*}h_j\,\chi(h_j)}{E}$. From the assumption that $h\,\chi(h)$ is concave, we know that $a'<b'$. We then have $\sum_{i\in I^*\setminus\hat{I}}h_i\,\chi(h_i)=a'D$ and $\sum_{j\in \hat{I}\setminus I^*}h_j\,\chi(h_j)=b'E$. The above inequalities are equivalent to
\begin{align}
\frac{D}{A}\frac{V-a'd}{m}&>\frac{E}{B}\frac{V-b'd}{n}\\
\frac{E}{C}\frac{V-b'd}{n}&<\frac{D}{C}\frac{V-a'd}{m},
\end{align}
Since $A<B$, $\frac{D}{m}>\frac{E}{n}$ and $a'<b'$, $(15)$ and $(16)$ are satisfied, which then implies $(13)$ and $(14)$. This completes the proof.
\end{proof}

The second condition in the proposition ensures that scientists with greater scientific capital $h_i$ obtain greater expected payoffs despite increasing opportunity cost per unit time. The existence of a unique risk dominant equilibrium for more general opportunity costs demonstrates the robustness of the basic model.

\subsection{Publishing Partial Results}
One of the defining characteristics of academic science is the expectation that scientists will share information \citep{partha1994toward}. Instead of securing (temporary) property rights over the exploitation of their discovery, as in a patenting regime, scientists compete for the recognition of their peers \citep{merton1973sociology}. In order to secure that recognition, scientists must share their discoveries. As a natural consequences of this Open Science norm, scientists grab for recognition early and often, publishing incremental results that fall well short of solving a given problem. Here we describe the conditions under which publication of partial results should obtain.

We assume that the problem $\mathcal{P}$ can be broken down into $M$ stages; this is sometimes called a treasure-hunt game \citep{banerjee2014re}. Each stage $m\in\{1,\cdots,M\}$ has value $V_m$ and difficulty $d_m$. Stage $m+1$ cannot be attacked until stage $m$ has been solved. We assume that each of these stages constitutes a publishable partial result.

At time $0$, all the scientists in $I$ simultaneously decide whether to participate in stage $1$. Scientists can participate in a later stage even without participation in an earlier stage. When scientist $i$ decides to participate in stage $m$, she pays a fixed cost $F_m$ and an opportunity cost of $h_i \, x$ per unit time\footnote{It may seem unrealistic that we do not require scientists to pay the fixed cost for earlier stages if they do not participate in those stages. We argue that, in most cases, the size of the fixed cost is driven by scientists' \emph{uncertainty} about the appropriate line of attack; they must prepare to battle on many fronts. Once the way forward is known, the fixed costs will be rapidly driven down by standardization and routinization, as the requisite technologies are black-boxed \citep{latour1987science}.}. A scientist is not required to participate in a subsequent stage $m+1$ after participating in stage $m$, nor will it always be optimal for her to do so. 

When a scientist solves a stage $m$, she faces an interesting dilemma. She can publish the solution immediately, guaranteeing that she receives credit for it but allowing her competitors to catch up.  Alternatively, she can keep her results to herself and begin working on the next stage, gambling the credit that she could have earned for the first stage in order to improve her chances of winning the race to solve the next stage.  Formally, when a solution to $m$ is published, every scientist in $I$ knows the solution (whether or not they participated in the stage). They then decide simultaneously whether to participate in stage $m+1$. If a scientist solves stage $m+1$ but chooses not to publish, she can participate in stage $m+2$ earlier than others. The value of the stage $V_m$ is awarded to the first scientist to \emph{publish} the solution, irrespective of who was first to \emph{solve} it.

We are interested here in the institution of Open Science, characterized by strategy profiles that lead to immediate publication of any partial result. Given such strategy profiles, individual strategies that abandon a stage mid-stream are never optimal and every participating scientist of every stage will always publish the solution with no delay. We call such a strategy profile a \textit{public sharing strategy profile (PSSP)}. When a PSSP constitutes an equilibrium (in the perfect Bayesian sense), we denote it as a \textit{public sharing equilibrium (PSE)} \citep{banerjee2014re}.

For stage $m$, let $I_m^*=\{1,\cdots,i_m^*\}$ where $i_m^*=\max\{i:V\frac{h_i}{\sum_{j=1}^i h_j}-h_i\,x\frac{d_m}{\sum_{j=1}^i h_j}-F_m\geq 0\}$. Let $h_{I_m^*-i}=\sum_{j\in I_m^*,j\neq i}h_j$. As before, scientists are labeled in descending order of scientific capital $h_i$: $h_1\geq h_2\geq\cdots\geq h_N$.

We prove the following proposition:

\begin{prop}
Suppose that the following condition is satisfied for every pair of stages $(m,m')$ such that $m$ precedes $m'$:
\begin{align*}
\text{For any $i$, }\frac{V_m h_{I^*_m-i}}{d_m}-\frac{h_i}{d_{m'}}\times\frac{V_{m'}h_{I^*_{m'}-i}+xd_{m'}h_i}{h_i+h_{I^*_{m'}-i}}\geq 0.
\end{align*}
Then the following strategy profile, denoted $\sigma^*$, constitutes a PSE: the scientists play the PSSP in which scientist $i$ always participates in stage $m$ if and only if $i\in I^*_m$; never abandons a stage mid-stream; and always publishes the solution to a stage without delay. In every stage, the participation decisions in $\sigma^*$ are risk dominant.
\end{prop}

\begin{proof}
Suppose that every agent other than $i$ chooses the PSSP strategy described above: agent $j$ participates if $j \in I^*_m$, never abandons a stage, and always publishes the solution to a stage without delay. First, we show that given any participation decision of $i$, publishing instantly is a best response. The proof of this part, which is given below, generalizes the argument by \cite{banerjee2014re} to an environment with the opportunity cost of time. We adapt the language of \cite{banerjee2014re} to our setting, but follow their argument and notation closely.

For notational convenience, we label the $m$ stages in reverse order, such that the last stage is labeled $1$ and the first is labeled $m$. Let $G_{k,l}^i$, $0\leq k\leq l\leq m$, be the subgame where agent $i$ is working on stage $k$, and all other agents are working on stage $l$. Note that because of the property of the exponential distribution, it does not matter how much time has elapsed for each agent at her current stage. Consider two subgames $G_{k,l}^i$ and $G_{p,q}^i$; if $p<k$, or $p=k$ and $q<l$,

denote the subgame $G_{p,q}^i$ as \emph{smaller} than $G_{k,l}^i$. We let $R(G_{k,l}^i)$ denote the expected payoff earned by agent $i$ in subgame $G_{k,l}^i$ (this corresponds to the ``expected reward" in \cite{banerjee2014re}). Finally, recall that in a subgame $G_{k,l}^i$ with $k<l$, agent $i$ is already at stage $k$, while other participating agents are ``behind" in the race and must work on an earlier stage $l$. We label the stages $\{l,l-1,\cdots,k+1\}$ (i.e., the stages between $l$ and $k$, including $l$) as the \emph{captive stages} of agent $i$: these are the stages that $i$ has solved and other agents have not. We prove the following claim: under the condition specified above, PSSP is a best response for agent $i$ in every subgame $G_{k,l}^i$, $0\leq k\leq l\leq m$. 

Like \cite{banerjee2014re}, we prove the claim by a joint induction on $k$ and $l$. The simplest case is one in which $i$ has solved every stage, i.e., the subgame $G_{0,l}^i$ where $i$ is already at stage $0$ and her competitors are at some stage $l$ equal to or earlier than $0$. It is obvious that $i$'s best response is to publish immediately, following the public sharing strategy profile (PSSP). In any other subgame, $i$ has solved some but not all of the stages and is either ahead of her competitors ($k<l$, Case I) or tied with them ($k=l$, Case II). We analyze each case in turn.

Case I: $1\leq k<l$. Consider a subgame $G_{k,l}^i$ and denote its \emph{starting time} as $t=0$ (in other words, at $t = 0$, $i$ knows the solution to all stages up to $k$ and her participating competitors know the solution to all stages up to $l$). Assume as our induction hypothesis that PSSP is a best response for agent $i$ in every smaller game, i.e., for every subgame $G_{p,q}^i$ in which $p<k$, or $p=k$ and $q<l$. Consider the case in which agent $i$ instantaneously publishes the solutions to her captive stages at $t=0$, i.e., in which she follows the PSSP. Upon publication of these solutions, $G_{k,l}^i$ reduces to the subgame in which all participating scientists start at the same stage $k$, i.e., to the smaller subgame $G_{k,k}^i$. By assumption, PSSP is the best response for this smaller subgame. Now consider the case in which $i$ deviates from PSSP, keeping her solutions secret for a finite amount of time. After an interval $\tau_1$ we transition to a different subgame for one of two reasons: either agent $i$ solves stage $k$ (so we transition to the smaller subgame $G_{k-1,l}^i$, in which $i$ will play PSSP as her best response by assumption) or one of her competitors solves stage $l$ and discloses the solution (so we transition to the smaller subgame $G_{k,l-1}^i$, in which $i$ will play PSSP as her best response by assumption). Recall that playing PSSP for $G_{k,l}^i$ would require immediate publication of all captive stages as soon as the game begins. We hence consider alternative strategies with the following structure: $i$ commits to publishing the solution to captive stages after some fixed delay $\tau>0$, or publishes the solutions at $\tau_1$, whichever comes first. We thus need to show that $i$ maximizes her expected payoff in the subgame $G_{k,l}^i$ by choosing $\tau = 0$, i.e., by playing the PSSP. We denote the expected payoff to agent $i$ of delay $\tau$ in subgame $G_{k,l}^i$ as $U_{k,l}^i(\tau)$. The expected payoff to the PSSP (i.e., $\tau = 0$) is 
\begin{align*}
U_{k,l}^i(0)=\sum_{u=k+1}^l V_u + R(G_{k,k}^i)=V_l+\Delta+\frac{h_i}{h_i+h_{I^*_k-i}}(V_k-dx)+R(G_{k-1,k-1}^i),
\end{align*}
where we define $\Delta=\sum_{j=k+1}^{l-1}V_j$. This corresponds to the payoffs for all stages $l$ to $k+1$ (denoted $\Delta$) plus the expected payoff in subgame $G_{k,k}^i$. On the other hand for $\tau>0$, if there were no opportunity cost of time, delaying publication until $t=\tau$ yields an expected payoff of
\begin{align*}
&Prob(\tau_1\geq\tau)(\sum_{u=k+1}^l V_u+R(G_{k,k}^i))\\
&+Prob(\tau_1<\tau)(\frac{\frac{h_i}{d_k}}{\frac{h_i}{d_k}+\frac{h_{I^*_l-i}}{d_l}}R(G_{k-1,l}^i)+\frac{\frac{h_{I^*_l-i}}{d_l}}{\frac{h_i}{d_k}+\frac{h_{I^*_l-i}}{d_l}}R(G_{k,l-1}^i)),
\end{align*}
where $\sum_{u=k+1}^l V_u+R(G_{k,k}^i)$ is $i$'s continuation payoff in the event that no competitor solves stage $l$ and $i$ does not solve stage $k$ by time $t=\tau$, and $\frac{\frac{h_i}{d_k}}{\frac{h_i}{d_k}+\frac{h_{I^*_l-i}}{d_l}}R(G_{k-1,l}^i)+\frac{\frac{h_{I^*_l-i}}{d_l}}{\frac{h_i}{d_k}+\frac{h_{I^*_l-i}}{d_l}}R(G_{k,l-1}^i)$ is $i$'s continuation payoff in the event that $i$ solves stage $k$ or some competitor solves stage $l$  by $t=\tau$. With the opportunity cost of time, $i$'s expected payoff must be lower, and hence we have

\begin{align*}
U_{k,l}^i(\tau)\leq &Prob(\tau_1\geq\tau)(\sum_{u=k+1}^l V_u+R(G_{k,k}^i))\\
&+Prob(\tau_1<\tau)(\frac{\frac{h_i}{d_k}}{\frac{h_i}{d_k}+\frac{h_{I^*_l-i}}{d_l}}R(G_{k-1,l}^i)+\frac{\frac{h_{I^*_l-i}}{d_l}}{\frac{h_i}{d_k}+\frac{h_{I^*_l-i}}{d_l}}R(G_{k,l-1}^i)).
\end{align*}

We can rewrite the bounds on the expected payoff $U_{k,l}^i(\tau)$ in terms similar to the expected payoff of the PSSP, $U_{k,l}^i(0)$:
\begin{align*}
U_{k,l}^i(\tau)\leq &Prob(\tau_1\geq\tau)(V_l+\Delta+\frac{h_i}{h_i+h_{I^*_k-i}}(V_k-dx)+R(G_{k-1,k-1}^i))\\
&+Prob(\tau_1<\tau)\frac{\frac{h_i}{d_k}}{\frac{h_i}{d_k}+\frac{h_{I^*_l-i}}{d_l}}(V_l+\Delta+V_k+R(G_{k-1,k-1}^i))\\
&+Prob(\tau_1<\tau)\frac{\frac{h_{I^*_l-i}}{d_l}}{\frac{h_i}{d_k}+\frac{h_{I^*_l-i}}{d_l}}(\Delta+\frac{h_i}{h_i+h_{I^*_k-i}}(V_k-dx)+R(G_{k-1,k-1}^i)).
\end{align*}
To compare this payoff with the expected payoff of publishing instantaneously, we subtract the above expression from $U_{k,l}^i(0)$ and have
\begin{align*}
U_{k,l}^i(0)-U_{k,l}^i(\tau)\geq \frac{Prob(\tau_1<\tau)}{\frac{h_i}{d_k}+\frac{h_{I^*_l-i}}{d_l}}(\frac{V_l h_{I^*_l-i}}{d_l}-\frac{h_i}{d_{k}}\frac{V_{k}h_{I^*_{k}-i}+xd_{k}h_i}{h_i+h_{I^*_{k}-i}}).
\end{align*}
From the condition in the proposition, we have $U_{k,l}^i(0)-U_{k,l}^i(\tau)\geq 0$ for every $\tau>0$. Thus setting $\tau=0$ maximizes the expected payoff.

Case II: $k=l$. Assume as our induction hypothesis that PSSP is a best response for agent $i$ in every smaller game. In this case, the game will transition to either $G_{k-1,k}^i$ (if $i$ solves stage $k$) or to $G_{k-1,k-1}^i$ (if some competitor solves stage $l$). In each case, PSSP is a best response for $i$ from the induction hypothesis. This completes the induction and hence proves the first part.

Next, we show that the \emph{participation} decision for $i$ as described in the proposition is indeed optimal. Given instantaneous publication by each agent, the whole game can be treated as $M$ independent games, one for each stage, with an identical payoff structure to the \emph{basic game} defined in Section 2. By Proposition 2, we know that for every such game, $I^*_m$ is an equilibrium for participation. Hence, given that every other scientist $j$ participates in stage $m$ if and only if $j\in I^*_m$, it is optimal for $i$ to use the analogous strategy: participate in stage $m$ if and only if $i\in I^*_m$. Therefore, $\sigma^*$ is a PSE. Finally, Proposition 2 guarantees that the participation decisions in each stage are risk dominant. This completes the proof.

\end{proof}

The basic intuition behind the PSE is as follows: A scientist will publish her solution to any stage immediately when (1) the current solved stage is relatively valuable or easy; (2) there are many competitors and/or competitors with high scientific capital; (3) she has low scientific capital herself or (4) the opportunity cost of time is low. Cases (1) - (3) refer to the situation where a bird in the hand is better than two in the bush: the scientist is incented to publish immediately and secure the reward before strong competitors catch up, which is a relatively likely event. In (4), both the expected payoff of publishing immediately and the expected payoff of holding off increase in the opportunity cost of time, but the increment of the former is larger.

\begin{exmp}
A species of Mycoplasma bacterium previously endemic in migratory songbirds begins to infect domestic ducks and geese.  Two microbiologists, Daniela and Elton, are poised to figure out how this host shift has occurred. The research process involves two stages: first, large-scale whole genome sequencing to determine what mutation or mutations made the shift possible; and second, the molecular biology necessary to work out the biochemical mechanism by which the mutations act. Daniela has scientific capital $h_1=10$ and Elton has scientific capital $h_2=6$. The two stages offer rewards $V_1=60$ and $V_2=30$ respectively to the scientist who first publishes the solution, and both have the same difficulty level $d=15$. The opportunity cost per unit scientific capital per unit time for both scientists is $x=1$, while the fixed costs for the two stages are $F_1=1$ and $F_2=5$. We first obtain $I^*_1=I^*_2=\{Daniela,Elton\}$. It is easy to check that the condition in Proposition 4 is satisfied for both scientists. Hence, the only PSE of this game is that both scientists participate in each stage and upon solving either stage each would immediately publish the results.

But now suppose that fixed costs are slightly higher for the second stage: $F_2=6$. Again both Daniela and Elton will opt to pursue the initial (and highly rewarding) genomic sequencing stage. Each will publish the results of this stage immediately. But only Daniela will pursue the mechanistic molecular biology in the second stage. Thus for this scenario once again $I^*_1=\{Daniela,Elton\}$, but now $I^*_2=\{Daniela\}$. It is straightforward to check that the condition in Proposition 4 is satisfied with the new $I^*_1,I^*_2$. Hence, the only PSE of this game is that both Daniela and Elton participate in stage 1, but only Daniela participates in stage 2.
\end{exmp}

The derivation of the sufficient condition for PSE is similar in part to the argument by \cite{banerjee2014re}, but our result differs from theirs in a fundamental aspect. In our model, time spent on each stage is costly and scientists are free to opt in or out for each stage, while in the Banerjee et al. model the researchers must continue to participate in each stage until the final one is solved. Consequently, the size and membership of the set of participating scientists may fluctuate across stages in our PSE $\sigma^*$: more scientists will choose to work on easier stages, while only the best few will stay for the harder stages.
The speed at which different stages are solved may be affected differentially as well.  The harder the stage, the larger the number of scientists who opt out and thus the longer its solution is delayed compared to the case of compulsory participation. From a social welfare point of view, this raises a further question concerning the design of institutions: should the institutional arrangement reflected in this multi-stage framework and commonly known as ``Open Science'' always be promoted, or is it sometimes more efficient to implement a scheme where scientists must participate in either every stage or none at all? We will answer this question in the next section.

\subsection{Social Welfare and Efficiency}

We consider social welfare as follows: once the problem $\mathcal{P}$ is solved, it generates a constant flow of benefit to society of $\hat{V}>0$ per unit time\footnote{
\citet{partha1994toward} assert that society values ``additions to the stock of public knowledge'' (which they dub Science) as distinct from ``adding to the stream of rents that may be derived from possession of...private knowledge'' (which they dub Technology). Based on this analytic distinction, scholars often treat time differently in models of Science versus models of Technology. According to this view, both society and individual technologists care about the timing of benefit streams flowing from Technology; hence time is often treated explicitly in these models. But in the case of Science, individual discoveries are modeled as benefit shocks and society is modeled as indifferent to the time required to add to the stock of knowledge. On this assumption, \emph{scientists} may care a great deal about time-to-discovery, but this time sensitivity isn't typically modeled. We will argue that this analytic distinction and the associated modeling choices distort important aspects of the social function of science. Most results in basic science generate ongoing benefit flows, and in many cases society cares deeply about the time-to-solution for a problem in basic science. Thus we account directly for time in the present analysis. }. These benefits could be largely cognitive (e.g., knowing that gravitational waves exist) or eminently practical (e.g., reducing deaths from AIDS). Society discounts benefit and cost by rate $r$, that is, the present value of receiving a benefit $V$ at time $t$ in the future is equal to $V \,e^{-r t}$. 

Finally, for publicly funded research, society ultimately bears some of the fixed costs and therefore there is a social fixed cost $\hat{F}$ per participating scientist. Because scientists respond to individual incentives generated by the credit allocation mechanisms within the field, we do not assume that scientists' benefits and costs are precisely aligned with those of society \citep{haldane1926should}. We leave open the possibilities that (i) scientists may gain prestige, promotions, and profit by working on problems that are of little value to society and that (ii) problems of substantial importance to society may be undervalued within academia. We do assume, however, that once fixed costs are paid, the focal problem merits solving from society's perspective:  $\frac{\hat{V}}{r}-\hat{x}\,d>0$.  

We can now compute the expected social welfare when the participating group of scientists is $I'\in I$. The flow of benefits to society is equivalent to an instantaneous benefit of $\frac{\hat{V}}{r}$ at the moment the problem is solved. The expected social welfare is given by
\begin{align*}
W(I')=\frac{\hat{V}}{r}\frac{\sum_{j\in I'}h_j}{\sum_{j\in I'}h_j+dr}-\frac{\hat{x}d\sum_{j\in I'}h_j}{\sum_{j\in I'}h_j+dr}-|I'|\hat{F}
\end{align*}
where $|I'|$ is the cardinality of $I'$. Note that this expression can only achieve a positive value if $\frac{\hat{V}}{r}-\hat{x}\,d>0$. In this expression, the first term is the discounted expected social benefit, the second term is the discounted expected social opportunity cost, and the third term is the sum of the fixed costs. By observing that $\frac{\sum_{j\in I'}h_j}{\sum_{j\in I'}h_j+dr}$ is increasing in $\sum_{j\in I'}h_j$, we know that the socially optimal participating group, conditional on the group size, is the group consisting of scientists with the highest scientific capital. Hence, letting $H_i=\sum_{j=1}^i h_j$, it is more useful to focus on the expected social welfare from the group $\{1,\cdots,i\}$:
\begin{align*}
W(i)=\frac{\hat{V}}{r}\frac{H_i}{H_i+dr}-\frac{\hat{x}dH_i}{H_i+dr}-i\hat{F}.
\end{align*}

The value of social welfare goes to infinity as $r$ goes to $0$. However, the difference between social welfare for two different participating groups has a finite limit. For instance, consider another participating group of scientists $\{1,\cdots,i'\}$ with $H_{i'}$ units of scientific capital in total; the difference is given by
\begin{align*}
(\frac{\hat{V}}{r}-\hat{x}d)\frac{(H_i-H_{i'})dr}{(H_i+dr)(H_{i'}+dr)}-(i-i')\hat{F}.
\end{align*}
The limit as $r\rightarrow 0$ is equal to $\hat{V}d(H_i-H_{i'})-(i-i')\hat{F}$.

Let $I^e$ denote the socially optimal participating group, i.e., the group that maximizes the social welfare. We prove the following proposition:

\begin{prop}
The socially optimal participating group is unique: $I^e=\{1,\cdots,i^e\}$ where $i^e=\max\{i:W(i)-W(i-1)\geq 0\}$.
\end{prop}

\begin{proof}
It suffices to show that $W(i)-W(i-1)$ is decreasing in $i$. Note that
\begin{align*}
\frac{H_i}{H_i+dr}-\frac{H_{i-1}}{H_{i-1}+dr}=\frac{(H_i-H_{i-1})dr}{(H_i+dr)(H_{i-1}+dr)}=\frac{h_idr}{(H_i+dr)(H_{i-1}+dr)}.
\end{align*}
In this expression, $h_i$ is decreasing in $i$, and both $H_i$ and $H_{i-1}$ are increasing in $i$. Hence, the value of this expression is decreasing in $i$. Now we have
\begin{align*}
W(i)-W(i-1)=(\frac{\hat{V}}{r}-\hat{x}d)(\frac{H_i}{H_i+dr}-\frac{H_{i-1}}{H_{i-1}+dr})-\hat{F},
\end{align*}
which is a linear in $\frac{H_i}{H_i+dr}-\frac{H_{i-1}}{H_{i-1}+dr}$. Hence, $W(i)-W(i-1)$ is decreasing in $i$. This completes the proof.
\end{proof}

In this setting, both under-participation and over-participation are possible in general. Rearranging the terms in $W(i)-W(i-1)$, we can see that the socially optimal participating group is characterized by the maximum $i$ such that
\begin{align*}
\frac{\frac{\hat{V}}{r}-\hat{x}d}{\hat{F}}\times\frac{h_i}{(h_i+H_{i-1}+dr)(1+\frac{H_{i-1}}{dr})}\geq 1.
\end{align*}
The first term on the left describes the attractiveness of the problem to society as a benefit-to-cost ratio; the second term is scientist $i$'s marginal importance to society, given that scientists $1,\cdots,i-1$ are already participating.

Note that the condition we obtained for determining the risk dominant equilibrium takes a similar though somewhat simpler form. The unique risk dominant equilibrium is characterized by the maximum $i$ such that
\begin{align*}
\frac{V-{x}d}{{F}}\times\frac{h_i}{H_i}\geq 1.
\end{align*}
The first term on the left is how attractive the problem is to a scientist if she were to work on it alone; again it is a benefit-to-cost ratio. The second term is scientist $i$'s marginal importance, not to society in a cooperative sense, but to the competition. It is simply evaluated as her own chance of solving the problem first. 

Those studying the norms of science often assert that when a research problem can be decomposed into stages, it is socially efficient to publish the results of each stage as soon as they are known (e.g., Merton 1957b; David 2003). \nocite{merton1957social,david2003economic} By sharing partial results, the argument goes, the entire community of interested scientists can begin working on a subsequent stage as soon as the previous stage is solved. In the context of our model, we might expect that allowing the publication of partial results will always maximize social efficiency via the PSE. But this argument ignores the fact that scientists must decide whether to participate in a given stage at all. This decision is endogenous and depends on the information-sharing strategies of the other players. When partial progress sharing is allowed,  many researchers may choose to drop out of the more difficult stages and the expected time to solution may increase, to the detriment of society. We illustrate with an example.

\begin{exmp}
Consider the second case of Example 3, in which $F_2=6$. Once both stages are solved, it will be clear how to prevent further infection among domestic fowl, and thus society will enjoy a benefit of $\hat{V}=100$ per unit time. The benefit flow begins when and only when both stages are solved, and is discounted by rate $r=0.5$. The costs for society are assumed to be identical to those for the scientists: $\hat{x}=1$, $\hat{F}_1=1$ and $\hat{F}_2=6$. The general formula to compute social welfare is

\[
\mbox{social welfare} = E[\mbox{discounted benefit flow}] - E[\mbox{discounted opportunity cost}] -\hat{F}_1-\hat{F}_2.
\]

We have shown that with publication of partial results, the unique risk dominant PSE is that Daniela participates in both stages but Elton participates only in stage 1; and that they will both publish immediately when they have a solution to either stage. The time to solution of both stages is thus the sum of two exponential random variables, with rate $\frac{d}{h_1+h_2}$ and $\frac{d}{h_1}$ respectively. Hence the social welfare is equal to $51.03$.

Without publication of partial results, if either scientist wants to participate, he or she has to pay the sum $F_1+F_2$, and will receive a reward of $V_1+V_2=90$ if the first to solve both stages. For each scientist, the time to solution is now a gamma-distributed random variable with shape $2$ and scale $\frac{d}{h_i}$. We can show that the only equilibrium in this case is that both Daniela and Elton participate in the whole problem. In this case the social welfare is equal to $55.28$. This social welfare is higher than that obtained with publication of partial results because the cost of not sharing the solution to the first stage is more than outweighed by the benefits of having both scientists participate in both stages of the problem.
\end{exmp}

Here we see that when participation in each stage becomes part of a scientist's strategic decision, the scientist weighs the expected benefit and cost for each stage and only opts in for those where she receives a positive net expected payoff. Consequently, under-participation occurs for hard stages and the overall participation level may be lower compared to the case without publication of partial results. The expected time to solution for the problem as a whole becomes longer, resulting in lower social welfare. Note the concordance with recent experimental evidence \citep{boudreau2015open} showing that regimes \emph{without} disclosure of partial results lead to higher levels of participation.

It is useful to contrast our findings with existing theories on patent races. In patent races, a firm may disclose information \emph{before} it is able to secure a patent in order to signal its commitment (as a leader) or move the finish line (as a follower). Some authors have made the claim that this behavior is always socially desirable \citep{defraja1993,bargill2003}, but others have noted that once entry and exit are part of a strategic decision, sharing may not be optimal \citep{bhattacharya1983,parchomovsky2000,gill2008}. The key element leading to suboptimal exit in the latter models is risk preference: the more risk averse competitors are, the more likely the lagging ones will exit prematurely, which slows down overall research progress. In our model, a scientist may publish a partial result in order to secure the credit, and the exit decision depends on difficulty level, scientific capital, and value; when an intermediate stage is relatively hard or its value is relatively low, a less capable scientist opts out. As a result, the complete solution to the problem is delayed on average.

\section{Conclusion}

Big new problems arise in scientific communities through several processes. First, methodological or technological breakthroughs can facilitate new avenues of investigation; for example, massive improvements in the ability to sequence ancient DNA initiated a race to determine the contribution of archaic hominins such as Neanderthals and Denisovans to the modern human gene pool \citep{sankararaman2014genomic}. Second, theoretical advances can raise empirical possibilities that challenge scientists to develop novel instrumentation to verify their predictions. For example, Einstein's theory of general relativity predicted the existence of gravitational waves; after an initial delay, physicists spent decades trying numerous approaches for detecting them before the recent success of LIGO \citep{abbott2016observation}. Finally, nature herself confronts us with novel challenges of dramatic societal importance, like the current struggle to control the spread of Zika virus and ameliorate the harm that it is causing \citep{petersen2016zika}.

Our paper uses game theory to explore how the scientific community responds to these big problems. We make two contributions to this concrete question. First, our model robustly predicts that only the top scientists---those with the most scientific capital---will participate and compete for the solution of a big new problem. Second, we challenge the conventional assertion that Open Science with immediate sharing of partial results is always socially efficient. We show that this is not generically true: if competing in any or all stages of the race is itself a strategic choice, Open Science may slow down the ultimate solution of the problem by deterring participation. One major reason that our result differs from previous work is that we simultaneously model both information sharing decisions and problem choice, in that we allow the actors in our model to opt in or opt out of a particular research program. 

In doing so, we have proposed a flexible modeling framework for studying science, within which we can analyze both individual behavior and social efficiency. The framework can be readily extended in many directions. We offer a few concrete suggestions here: (1) How do scientists choose among multiple competing big problems? (2) What happens when society values the total stock of knowledge rather than a continuing flow of benefits? (3) How would scientists' behavior change if scientific credit were endogenized, so that the value of solving a problem depends on the number of individuals also trying to solve it? (4) How do scientists' decisions to enter or exit a race change under uncertainty about its difficulty? (5) What happens when stages of a problem can be solved in parallel rather than only in sequence? (6) In what contexts can scientists be incented to sequence their work on interrelated problems to minimize the total time-to-solution? (7) How do scientists behave in multi-period settings, in which  scientific capital is accumulated across periods? Each of these extensions increase the realism of the models, with direct implications for science policy. In the long run we aspire to make policy recommendations based upon such formal analysis rather than simple \emph{ad hoc} models or, more commonly, qualitative intuitions, hopes, and prejudices.

Efforts to increase the formal rigor of science policy recommendation are timely. Federal research budgets are stagnant or contracting, while the very public replication crisis \citep{open2015estimating,camerer2016evaluating} has shaken both lay and professional confidence in the basic institutions of science. One striking positive change is the flourishing literature on the science of science. Thanks to the unprecedented availability of large-scale electronic corpora \citep{evans2011metaknowledge}, the last 15 years have witnessed an explosion of large-scales---though often descriptive---quantitative analyses of scientific output \citep{rosvall2008maps, petersen2011quantitative, wang2013quantifying, uzzi2013atypical,foster_tradition_2015}. While these analyses have radically expanded our positive understanding of the operations of science, the policy prescriptions derived from these empirical projects have limited value. Most do not persuasively nail down causality, though \cite{budish2015firms} and \cite{azoulay2013matthew} are exceptions. In general, the field lacks the formal microfoundations needed for disciplined comparative institutional analysis. And such institutional analysis is necessary if we are to avoid unintended consequences from hasty or faddish policy changes.

We aim to provide such microfoundations. As we hint above, this modeling framework should be rich enough to treat choice, outside options, time, and partial publication in one setting, mirroring the complexity of scientific institutions and decision-making. The framework should also consistently distinguish between scientists' incentives and the social benefits of science, so that we can model a wide variety of situations. Here we address the most basic problem confronting scientists: research choice. We focus on the special and socially consequential case of ``big problems." 

Understanding scientific institutions is essential, if science and technology policy is to help rather than hinder scientists' capacity to solve specific problems and produce fundamental insights with unknown but often substantial spillover benefits. The recent empirical attention lavished on science is a step in the right direction. But observational and experimental studies can only get so far without sharpened theoretical tools to guide intuition, direct attention, and discipline policy implications. The multiplying crises of science constitute a novel challenge of dramatic social importance. Fixing science may be the next big problem.

\end{document}